\documentclass[a4paper,11pt]{article}
\usepackage[utf8]{inputenc}
\usepackage{amssymb}
\usepackage{amsmath}
\usepackage{amsthm}
\usepackage{dsfont}
\usepackage{mathrsfs}
\usepackage[T1]{fontenc}
\usepackage[ngerman,english]{babel}
\usepackage[inline]{enumitem}
\usepackage{graphicx}
\usepackage[all]{xy}
\usepackage{mathtools}

\numberwithin{equation}{section}

\DeclareMathOperator{\supp}{supp}

\DeclareMathOperator{\sgn}{sgn}

\newtheorem{theorem}{Theorem}

\newtheorem{lemma}{Lemma}
\newtheorem{corollary}{Consequence}
\theoremstyle{remark}
\newtheorem{remark}{Remark}

\theoremstyle{definition}

\newtheorem{assumption}{Assumptions}

\newcommand{\dda}{{\rm d}}
\newcommand{\dd}{\,\dda}
\newcommand{\ddd}[1]{\dd^{#1}}
\newcommand{\ee}{{\rm e}}

\begin{document}

\title{On contact interactions as limits of short-range potentials}

\author{Gerhard Br\"aunlich$^1$, Christian Hainzl$^1$ and Robert Seiringer$^2$
  \\ \\  $^1$Mathematical Institute, University of T{\"u}bingen \\ Auf der Morgenstelle 10, 72076 T\"ubingen, Germany \\ \\
  $^2$Institute of Science and Technology Austria\\ Am Campus 1, 3400 Klosterneuburg, Austria}

\date{\today}

\maketitle

\begin{abstract}
We reconsider the norm resolvent limit of $-\Delta + V_\ell$ with
$V_\ell$ tending to a point interaction in three dimensions. We are mainly interested in potentials $V_\ell$ modelling short range interactions
of cold atomic gases.
In order to ensure stability the interaction $V_\ell$ is
required to have a strong repulsive core, such that $\lim_{\ell \to 0}
\int V_\ell >0$. This situation is not covered in the previous
literature.
\end{abstract}

\section{Introduction}
\label{sec:introduction}

Quantum mechanical systems with contact interaction, or point
interaction, are treated extensively in the physics literature, in
connection with problems in atomic, nuclear and solid state physics.
In two and three dimensions such point interaction Hamiltonians have
to be defined carefully for the simple reason that the Dirac
$\delta$-function is not relatively form-bounded with respect to the
kinetic energy described by the Laplacian $-\Delta$. Mathematically
this can be overcome by removing the point of interaction from the
configuration space and extending $-\Delta$ to a self-adjoint
operator. This leads to a one-parameter family of extensions in two and
three dimensions. One way to pick out the physically relevant
extension is by approximating the contact interaction by a sequence of
corresponding short range potentials $V_\ell (x)$ such that the range
$\ell$ converges to zero, but the {\em scattering length} $a(V_\ell) $
has a finite limit $a\in \mathbb{R}$. The corresponding self-adjoint
extension is then uniquely determined by this $a$. 

The mathematical
analysis of such problems is extensively studied in the book of
Albeverio, Gesztesy, Hoegh-Krohn and Holden \cite{albeverio}.  Among
other things the authors show that $-\Delta + V_\ell$ converges in
norm resolvent sense to an appropriate self-adjoint extension of the
Laplacian determined by $a$. As one of their implicit assumptions the
$L^1$-norm of $V_\ell$ goes to zero in the limit $\ell \to 0$. This
assumption can be too restrictive for applications, however, as
explained in \cite{BHS}.

Indeed, one major area of physics where contact interactions play
a significant role are {\em cold atomic gases}, see e.g. \cite{Leggett, randeria, BHS}.
The corresponding BCS gap equation has a particularly simple form in this case.   However,
in order to prevent such a Fermi gas from collapsing and to ensure 
stability of matter, the contact interaction has to arise from
potentials $V_\ell$ which have a large repulsive core (such that
$\lim_{\ell \to 0} \int V_\ell > 0$) in addition to an attractive tail. The
strength of the attractive tail depends on the system under
consideration, ranging from a weakly interacting superfluid (where $-\Delta +
V_\ell \geq 0$) \cite{HS,HSch} to a strongly interacting 
gas of tightly bound fermion pairs (where $-\Delta + V_\ell$  typically has one
negative eigenvalue) \cite{randeria, HHSS, FHSS-micro_ginzburg_landau}.
 
Having such systems in mind, the present paper is dedicated to the
study of contact interactions arising as a limit of short range
potentials $V_\ell$ with large positive core and, in particular, a non-vanishing and positive integral $\int
V_\ell$ in the limit $\ell\to 0$. The simplest form of such a $V_\ell$ we can think of is depicted
in Figure~\ref{fig:potential}. 
We shall generalize a result of \cite{albeverio82,albeverio} and show that for $ \Im
(k) > 0$ 
$$ \frac 1{- \Delta + V_\ell - k^2} \stackrel{\ell\to 0}{\longrightarrow} \frac{1}{-\Delta-k^2} - \frac{4\pi}{a^{-1}+ik}|g_k\rangle
\langle g_k|
$$  in norm, where  $g_k(x) = \frac{1}{4\pi}\frac{\ee^{ik|x|}}{|x|}$, and $a = \lim_{\ell \to 0} a(V_\ell)$ is the limiting scattering length.
The main mathematical obstacle we have to overcome is the fact that
the corresponding Birman-Schwinger operators are ``very''
non-self-adjoint, which requires a refined analysis to get a hand on
the corresponding norms. 
One important ingredient of our analysis is a useful formula for the
scattering length $a(V_\ell)$ which was recently derived in \cite{HS-mu}.

\section{Main Results}
\label{sec:results}

In the following, we shall consider a family $(V_\ell)_{\ell>0}$ of
real-valued functions in $L^1(\mathbb{R}^3) \cap
L^{3/2}(\mathbb{R}^3)$.
We use the notation $V^{1/2}(x) = \sgn(x)|V(x)|^{1/2}$ and write
$V_\ell^\pm$ for the positive and negative part of the potential
$V_\ell$ in the decomposition
\begin{equation*}
  V_\ell = V_\ell^+ - V_\ell^-,\qquad \supp(V_\ell^+) \cap
  \supp(V_\ell^-) = \emptyset.
\end{equation*}
Further we will abbreviate
\begin{align*}
  J_\ell &= \big\{
  \begin{smallmatrix}
    1,& V_\ell \geq 0\\
    -1,& V_\ell <0,
  \end{smallmatrix}
  \qquad \textrm{i.e. } J_\ell(x) = \sgn\big(V_\ell(x)\big),\\
  X_\ell &= |V_\ell|^{1/2}\frac{1}{p^2}|V_\ell|^{1/2}, \qquad
  X_\ell^- = (V_\ell^-)^{1/2}\frac{1}{p^2}(V_\ell^-)^{1/2},
\end{align*}
so that the Birman-Schwinger operator reads
\begin{equation}\label{bso}
B_\ell := V_\ell^{1/2}\tfrac{1}{p^2}|V_\ell|^{1/2} = J_\ell X_\ell\,.
\end{equation}

For a given real-valued potential $V \in L^1(\mathbb{R}^3)
\cap L^{3/2}(\mathbb{R}^3)$, it was shown in \cite{HS-mu} that the
scattering length can be expressed via
\begin{equation}
  \label{eq:a}
  a(V) = \frac{1}{4\pi} \Bigl< |V|^{1/2}\Big| \tfrac{1}{1+V^{1/2}
    \frac{1}{p^2}|V|^{1/2}} V^{1/2}\Bigr>.
\end{equation}
This assumes that $1+V^{1/2}
    \frac{1}{p^2}|V|^{1/2}$ is invertible, otherwise $a(V)$ is infinite.

Throughout the  paper we will use the notation
\begin{equation*}
  f = O(g) \Leftrightarrow 0 \leq \limsup_{\ell \to 0}
  \left|\frac{f(\ell)}{g(\ell)}\right| < \infty.
\end{equation*}
For our main theorem we will need to make the following assumptions. 

\begin{assumption}
  \label{asm:assumption}
  \begin{enumerate}[label=(A\arabic*)]
  \item $(V_\ell)_{\ell>0} \in L^{3/2}(\mathbb{R}^3) \cap
    L^1\big(\mathbb{R}^3,(1+|x|^2)dx\big)$.    \label{ax:Lp}
  \item There are sequences $e_\ell, e_\ell^- \in \mathbb{R}$ such
    that $e_\ell \neq 0$, $\lim_{\ell\to 0} e_\ell = 0$, $e_\ell^- = O(e_\ell)$ and
    \begin{equation*}
      -\Delta + \lambda V_\ell
      \quad \textrm{and}\quad
      -\Delta - \lambda^- V_\ell^-
    \end{equation*}
    have non-degenerate zero-energy resonances for $\lambda = (1-e_\ell)^{-1}$ and
    $\lambda^- = (1-e_\ell^-)^{-1}$, respectively. 
    All other $\lambda, \lambda^- \in\mathbb{R}$ for which $-\Delta +
    \lambda V_\ell$ and $-\Delta - \lambda^- V_\ell^-$ have zero-energy 
    resonances are separated from $1$ by a gap of order $1$.
    \label{ax:e}
  \item $\|V_\ell\|_{L^1}$ is uniformly bounded in $\ell$ and
    $\|V_\ell^-\|_{L^1} = O(e_\ell)$,\label{ax:V}
  \item $\displaystyle \int_{\mathbb{R}^3} |V_\ell(x)|\,|x|^2 \ddd{3}x
    = O(e_\ell^2)$ and $\displaystyle \int_{\mathbb{R}^3}
    |V_\ell^-(x)|\,|x|^2 \ddd{3}x = O(e_\ell^3)$,\label{ax:V-x}
  \item the limit
    \begin{equation}\label{def:a}
      a = \lim_{\ell\to 0} a(V_\ell)
      = \frac{1}{4\pi}\lim_{\ell\to 0}
      \big\langle
      |V_\ell|^{1/2}\big| (1+B_\ell)^{-1}  V_\ell^{1/2} \big\rangle
    \end{equation}
    exists and is finite.\label{a5as}
\end{enumerate}
\end{assumption}

\begin{remark}
  Assumption~\ref{ax:e} can be reformulated in terms of the corresponding
  Birman-Schwinger operators. Recall that $-\Delta + \frac{1}{1-e} V$ has
  a zero-energy resonance if and only if $1+V^{1/2}\frac{1}{p^2}|V|^{1/2}$ has an eigenvalue $e$. 
  Therefore \ref{ax:e} is equivalent to the following assumption:

  \begin{enumerate}[label=(A\arabic*)]
    \item[\ref{ax:e}']
    The lowest eigenvalues $e_\ell$ and $e_\ell^-$ of the
     operators $1+J_\ell X_\ell$ and
    $1-X_\ell^-$, 
     respectively, are non-degenerate, 
    converge to $0$ as $\ell \to 0$, with $ e_\ell^-=O(e_\ell) $,  and all other eigenvalues
    are isolated from $0$ by a gap of order $1$.
  \end{enumerate}
The fact that $e_\ell \neq 0$ means that $1+B_\ell$ is invertible. For simplicity, we also assume that the limit in (\ref{def:a}) is finite. We expect our result to be true also for $a=\infty$, but the proof has to be suitably modified in this case. 
\end{remark}

We are now ready to state our main theorem.
\begin{theorem}
  \label{thm:convergence}
  Let $(V_\ell)_{\ell>0}$ be a family of real-valued functions
  satisfying Assumptions~\ref{asm:assumption}.
  Then, as $\ell \to 0$,
  \begin{equation}
    \label{eq:resolvent:2}
    \frac{1}{-\Delta + V_\ell - k^2} \rightarrow
    \frac{1}{-\Delta-k^2}
    - \frac{4\pi}{a^{-1}+ik}|g_k\rangle\langle g_k|,
    \qquad \Im\, k > 0\, , \ k\neq i/a
  \end{equation}
  in norm, where $g_k(x) = \frac{1}{4\pi}\frac{\ee^{ik|x|}}{|x|}$.
\end{theorem} 

\begin{remark}
  The simplest example of potentials satisfying Assumptions \ref{asm:assumption} 
  is shown in Figure~\ref{fig:potential}.
  \begin{figure}
    \centering
    \includegraphics[width=0.15\linewidth]{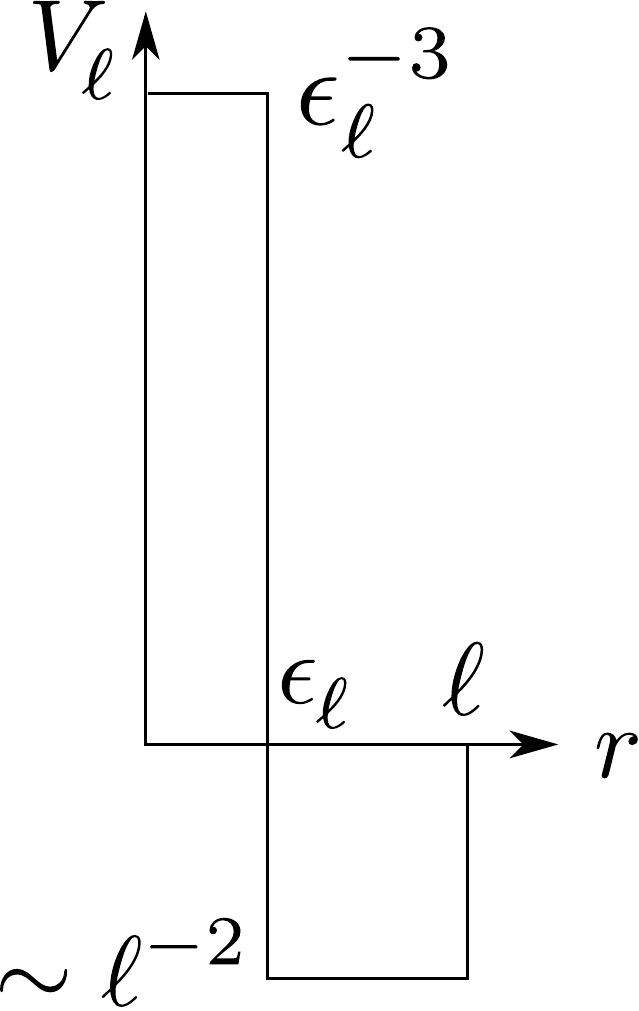}
    \caption{Example of a sequence  of potentials $V_\ell$.}
    \label{fig:potential}
  \end{figure}
  By simple calculations it is immediate to see that  \ref{ax:Lp}, \ref{ax:V} and
  \ref{ax:V-x} hold, with $e_\ell = O(\ell)$.  By fine tuning the strength of the negative part of $V_\ell$, it is
  possible to meet a resonance condition such that  \ref{ax:e} holds. The corresponding scattering length can be calculated explicitly in this case, verifying \ref{a5as}. (See \cite[Appendix]{BHS} for such explicit calculations in the case where $\epsilon_\ell \ll \ell$.)
  \end{remark}
  
\begin{remark}
In case $a=0$, the fraction $(a^{-1} + ik)^{-1}$ has to be interpreted as $0$.
\end{remark}  
  
\begin{remark}
  One consequence of Theorem~\ref{thm:convergence} is that in the case $0<a<\infty$ the smallest eigenvalue of 
  $-\Delta + V_\ell$ converges to $-\tfrac{1}{a^2}$, with the eigenfunction tending to $\sqrt{\frac {2\pi}a} g_{i/a}(x)$ in $L^2$. Moreover, all other eigenvalues necessarily tend to $0$. 
  
\end{remark}
\begin{remark}
  The resolvent on the right side of \eqref{eq:resolvent:2}
  belongs to a Hamiltonian of a special
  point interaction centered at the origin.
  More precisely, for $\theta\in[0,2\pi)$ let $\big(H_\theta, D(H_\theta)\big)$ be
  the self-adjoint extension of the kinetic energy
  Hamiltonian
  \begin{equation*}
    \left.-\Delta\right|_{H_0^{2,2}(\mathbb{R}^3\setminus \{0\})}
  \end{equation*}
  with the domain
  \begin{equation*}
    D(H_\theta) = H_0^{2,2}(\mathbb{R}^3\setminus \{0\}) \oplus
    \langle \psi_+ + \ee^{i\theta}\psi_- \rangle,\qquad
    \psi_\pm(x) = \frac{\ee^{i\sqrt{\pm i}|x|}}{4\pi |x|},\qquad
    x\in \mathbb{R}^3\setminus \{0\}, \Im(\sqrt{\pm i}) > 0,
  \end{equation*}
  such that $H_\theta(\psi_+ + \ee^{i\theta}\psi_-) = i\psi_+ -
  i\ee^{i\theta}\psi_-$.
  Then, if $\theta$ is chosen such that
  \begin{equation*}
    -a^{-1} = \cos(\pi/4)\big(\tan(\theta/2)-1\big),
  \end{equation*}
  we have by \cite[Theorem 1.1.2.]{albeverio} that
  \begin{equation}
    \label{eq:resolvent:delta}
    \frac{1}{H_\theta - k^2} = \frac{1}{p^2-k^2} - \frac{4\pi}{a^{-1}+ik}|g_k\rangle
    \langle g_k|.
  \end{equation}
 Hence Theorem~\ref{thm:convergence} implies that
  the operator $-\Delta + V_\ell$
  converges to $(H_\theta, D(H_\theta))$ in norm resolvent sense.
\end{remark}
\begin{remark}
  Let us explain one of the main difficulties arising from potentials with large
  $L^1$-core compared to the situation treated in \cite{albeverio},
  where the $L^{3/2}$-norm of $V_\ell$ is uniformly bounded and hence the $L^1$-norm
  tends to zero.
One of the necessary tasks in the proof of Theorem
  \ref{thm:convergence} is to
  bound the inverse of operator $1+B_\ell$, where $B_\ell$ denotes the Birman-Schwinger operator defined in (\ref{bso}). 
One way to bound the norm of this non-self-adjoint operator is
to use the
identity
$$\frac 1 {1+ B_\ell}  = 1 - J_\ell X_\ell^{1/2} \frac 1 {1+ X_\ell^{1/2}J_\ell X_\ell^{1/2}}  X_\ell^{1/2} ,$$
which implies for its norm
$$ \left \| \frac 1 {1+ B_\ell}  \right \| \leq 1 + \|X_\ell\| \left \|   \frac 1 {1+ X_\ell^{1/2}J_\ell X_\ell^{1/2}}  \right \|.
$$ 
The Hardy-Littlewood-Sobolev inequality implies that $\| X_\ell\|$ can be bounded by a constant times $ \|V_\ell\|_{L^{3/2}}$, and with $ X_\ell^{1/2}J_\ell
X_\ell^{1/2}$ being isospectral to $B_\ell=J_\ell X_\ell$ we obtain that
\begin{equation}
  \label{opin}  \left \| \frac 1 {1+ B_\ell} \right \|
  \leq 1 +C \frac 1{|e_\ell|} \|V_\ell\|_{L^{3/2}} \,.
\end{equation}
This shows that $\left \| (1+ B_\ell)^{-1} \right \| \leq O(|e_\ell|^{-1})$ for sequences $V_\ell$ used in \cite{albeverio}, which turns out to be sufficient. 
However, in our present situation we are dealing with potentials
$V_\ell$ with strong repulsive core satisfying Assumptions~\ref{asm:assumption}, where 
the corresponding  $L^{3/2}$ norm diverges and  typically is of the order of 
$O(1/|e_\ell|)$.  Hence the inequality \eqref{opin} only implies a
bound of the form
 $$\left \| \frac 1 {1+ B_\ell}  \right \| \leq O\big(|e_\ell|^{-2}\big)\,, $$
 which is not good enough for our purpose. To this aim we have to
 perform a more refined analysis. 
 \end{remark}
\begin{remark}
For related work on two-scale limits in one-dimensional systems, see, e.g., \cite{Exner, Gol}.
\end{remark}

The following lemma turns out to be very useful in the proof of Theorem~\ref{thm:convergence}.

\begin{lemma}
  \label{lemma:BS-estimate}
  Let $V = V_+ - V_-$, where $V_-, V_+\geq 0$ have disjoint
  support. Denote $J = \big\{
  \begin{smallmatrix}
    1,& V \geq 0\\
    -1,& V <0
  \end{smallmatrix}$, $\quad$ $X = |V|^{1/2}\frac{1}{p^2}|V|^{1/2}$
  and $X_\pm = V_\pm^{1/2}\frac{1}{p^2}V_\pm^{1/2}$.  Then for any
  $\phi \in L^2(\mathbb{R}^3)$, we have
  \begin{equation}
    \label{eq:B-S-inequality}
    \sqrt{2}\  \|\phi\|_{L^2} \| (J+X)\phi \|_{L^2}
    \geq \langle \phi| (X_+ +1 - X_-)\phi\rangle.
  \end{equation}
\end{lemma}
\begin{proof}[Proof of Lemma~\ref{lemma:BS-estimate}]
  Decompose $\phi = \phi_+ + \phi_-$, such that $\supp(\phi_-)
  \subseteq \supp(V_-)$ and $\supp(\phi_+) \cap \supp(V_-) =
  \emptyset$.  By applying twice the Cauchy-Schwarz inequality, 
  \begin{align*}
    \|(J+X)\phi\|_{L^2} \|\phi_+\|_{L^2} &\geq \Re \langle \phi_+|
    (J+X) \phi \rangle = \langle \phi_+| (1+X_+) \phi_+\rangle + \Re
    \langle \phi_+| V_+^{1/2}\tfrac{1}{p^2} V_-^{1/2} \phi_-\rangle,
    \\
    \|(J+X)\phi\|_{L^2} \|\phi_-\|_{L^2} &\geq \Re \langle (J+X)
    \phi|\ -\phi_-\rangle = \langle \phi_-| (1-X_-) \phi_-\rangle
    -\Re\langle \phi_+| V_+^{1/2}\tfrac{1}{p^2} V_-^{1/2}
    \phi_-\rangle.
  \end{align*}
  We add the two inequalities to get rid of the cross terms
  \begin{equation}
    \label{eq:BS-estimate:2}
    \|(J+X)\phi \|_{L^2}\ \big(\|\phi_+\|_{L^2} + \|\phi_-\|_{L^2}\big)
    \geq
    \langle \phi_+| (1+X_+) \phi_+\rangle
    + \langle \phi_-| (1-X_-) \phi_-\rangle
    = \langle \phi| (X_+ +1 - X_-)\phi\rangle.
  \end{equation}
  Finally, we use that $\|\phi_+\|_{L^2} + \|\phi_-\|_{L^2} \leq
  \sqrt{2} \|\phi\|_{L^2}$, which finishes the proof.
\end{proof}

One difficulty in proving Theorem~\ref{thm:convergence} is that the
 operator $1+B_\ell$ is not self-adjoint and
the norm of its inverse cannot be controlled by the spectrum. One consequence of 
 Lemma~\ref{lemma:BS-estimate}
and  our assumptions is that the norm of $(1+B_\ell)^{-1}$  diverges like
$\frac{1}{e_\ell}$. The following statement identifies the divergent term
in terms of the projection onto the eigenvector to the lowest
eigenvalue of the Birman Schwinger operator.

\begin{corollary}
  \label{corollary:projection-bound}
  Let $(V_\ell)_{\ell>0}$
  satisfy \ref{ax:Lp}-\ref{ax:V-x} in Assumptions~\ref{asm:assumption}.
  Then the operator
  \begin{equation*}
    (1+B_\ell)^{-1} (1-P_\ell)
  \end{equation*}
  is uniformly bounded in $\ell$, where
  \begin{equation}
    \label{eq:P}
    P_\ell = \frac{1}{\langle J_\ell
      \phi_\ell|\phi_\ell\rangle}|\phi_\ell\rangle \langle J_\ell
    \phi_\ell|
  \end{equation}
  with $\phi_\ell$ the eigenvector to the
  eigenvalue $e_\ell$ of $1+ B_\ell$.
\end{corollary}

Another consequence of Lemma~\ref{lemma:BS-estimate} is the following
set of relations, which the proof of Theorem~\ref{thm:convergence}
heavily relies on.

\begin{corollary}
  Let $(V_\ell)_{\ell>0}$ be a family of real-valued functions which
  satisfy \ref{ax:Lp}-\ref{ax:V-x} in Assumptions \ref{asm:assumption}.
  Then
  \label{corollary:V-phi}
  \begin{enumerate}[label=(\roman*)]
  \item $\int_{\mathbb{R}^3} |x|\,|V_\ell(x)| \ddd{3}x =
    O(e_\ell)$, \label{corollary:V-phi:4}
  \item $\langle J_\ell \phi_\ell|\phi_\ell\rangle = -1 +
    O(e_\ell)$, \label{corollary:V-phi:1}
  \item $\langle |V_\ell|^{1/2}||\phi_\ell|\rangle =
    O(|e_\ell|^{1/2})$, \label{corollary:V-phi:2}
  \item $\int_{\mathbb{R}^3} |x|\,|V_\ell(x)|^{1/2}\,|\phi_\ell(x)|
    \ddd{3}x = O(|e_\ell|^{3/2})$.
    \label{corollary:V-phi:3}
  \end{enumerate}
\end{corollary}

The proof of these facts will be given Section~\ref{sec:proof-cor}.

\section{Proof  of Theorem~\ref{thm:convergence}}
\label{sec:proofs}
Let $k\in\mathbb{C}$ with $\Im\, k > 0$. 
Following the strategy in \cite{albeverio} our starting point is the identity 
  \begin{equation}
    \label{eq:resolvent}
    \frac{1}{p^2+V_\ell - k^2}
    = \frac{1}{p^2-k^2}
    - \underbracket[0.5pt][2pt]{\frac{1}{p^2-k^2}
      |V_\ell|^{1/2}}_{\textrm{\textcircled{\scriptsize 1}}} \underbracket[0.5pt][2pt]{\frac{1}{1+V_\ell^{1/2}
        \frac{1}{p^2-k^2}|V_\ell|^{1/2}}}_{\textrm{\textcircled{\scriptsize 2}}} \underbracket[0.5pt][2pt]{V_\ell^{1/2} \frac{1}{p^2-k^2}}_{\textrm{\textcircled{\scriptsize 3}}}.
  \end{equation}
  Note that the operators \textcircled{\scriptsize 1} and
  \textcircled{\scriptsize 3} are uniformly bounded in $\ell$. In fact,
  \begin{equation*}
    \Bigl\|\frac{1}{p^2-k^2} |V_\ell|^{1/2}\Bigr\| \leq
    \Bigl\|\frac{1}{p^2-k^2} |V_\ell|^{1/2}\Bigr\|_2 =
    \bigl\|\frac{1}{p^2-k^2}\bigr\|_{L^2} \| V_\ell^{1/2}\|_{L^2},
  \end{equation*}
  which is uniformly bounded due to our assumptions on $V_\ell$.  We will first show in Lemma~\ref{newlem} that $\frac{1}{p^2-k^2} |V_\ell|^{1/2}$ is, up to small errors of $O(|e_\ell|^{1/2})$,   equal to the rank one operator $ |g_k\rangle
  \langle |V_\ell|^{1/2} | $. Together with the formula (\ref{eq:a}) for the scattering length this will lead us finally
  to \eqref{eq:resolvent:2}. 
  
  \begin{lemma}\label{newlem}
     \begin{subequations}
      \begin{align}
        \label{eq:convergence:3:a}
        \Bigl\| \Bigl(\tfrac{1}{p^2-k^2} |V_\ell|^{1/2} - |g_k\rangle
        \bigl< |V_\ell|^{1/2}\bigr|\Bigr)
        \tfrac{1}{1+B_\ell}|V_\ell^{1/2}\rangle \Bigr\|_{L^2} &= O(|e_\ell|^{1/2})\\
        \label{eq:convergence:3:b}
        \Bigl\| \Bigl(\tfrac{1}{p^2-k^2} |V_\ell|^{1/2} - |g_k\rangle
        \bigl< |V_\ell|^{1/2}\bigr|\Bigr) \tfrac{1}{1+B_\ell}V_\ell^{1/2} \tfrac{1}{p^2-k^2}\Bigr\| &=
        O(|e_\ell|^{1/2}).
      \end{align}
    \end{subequations}
  \end{lemma}
  
  \begin{proof}
We make use of  the
    decomposition
    \begin{equation}
      \label{eq:BS-decomposition}
      \frac{1}{1+B_\ell} = \frac{1}{e_\ell} P_\ell + \frac{1}{1+B_\ell}(1-P_\ell),
    \end{equation}
    where we first treat the contribution of the second summand to
    \eqref{eq:convergence:3:a}  and \eqref{eq:convergence:3:b}. This is the easier one thanks to
    Consequence~\ref{corollary:projection-bound}, which tells us that $ \frac{1}{1+B_\ell}(1-P_\ell)$ is uniformly bounded in $\ell$. 
Note that the integral kernel of  the operator $\tfrac{1}{p^2-k^2} |V_\ell|^{1/2} - |g_k\rangle
        \bigl< |V_\ell|^{1/2}|$
    in \eqref{eq:convergence:3:a} and \eqref{eq:convergence:3:b} is
    given by the expression
    \begin{equation*}
      \big(g_k(x-y) - g_k(x)\big) |V_\ell|^{1/2}(y)\,.
    \end{equation*} 
  An explicit computation shows that 
    \begin{equation*}
      F_k(y) :=  \int_{\mathbb{R}^3} \frac{\ee^{ik  |x|}}{|x|}
      \frac{\ee^{-i \bar k |x-y|}}{|x-y|} \ddd{3}x
      = \frac{2\pi}{\Im(k)}\ee^{-\Im(k)|y|}
        \frac{\sin(\Re(k)|y|)}{\Re(k)|y|}\,.
    \end{equation*}
From this one easily obtains the bound 
      \begin{equation}
      \label{eq:G-difference}
      \int_{\mathbb{R}^3} \big| g_k(x-y) - g_k(x) \big|^2 \ddd{3}x = \frac 2{(4\pi)^{2}} \left( F_{k}(0) - F_{k}(y) \right)  \leq \frac 1
      {4\pi} \left(1+ \frac{|\Re (k)|}{2|\Im(k)|}\right) |y|\,.
    \end{equation}
Hence we infer from Consequence~\ref{corollary:V-phi}\ref{corollary:V-phi:4} that the Hilbert-Schmidt norm of this operator is bounded as    
    \begin{equation}
      \label{eq:convergence:GV}
      \Bigl\|\tfrac{1}{p^2-k^2} |V_\ell|^{1/2} - |g_k\rangle \bigl<
      |V_\ell|^{1/2}\bigr| \Bigr\|_2
      \leq \left[
        \tfrac 1{4\pi} \Bigl(1+ \tfrac{|\Re (k)|}{2|\Im(k)|}\Bigr)
        \int_{\mathbb{R}^3} |V_\ell(x)||x| \ddd{3}x
      \right]^{1/2}
      = O(|e_\ell|^{1/2}).
    \end{equation}
 Thus the contribution of the last term in (\ref{eq:BS-decomposition}) to (\ref{eq:convergence:3:a}) gives a term of order $|e_\ell|^{1/2}$. 
   With $V_\ell^{1/2}\tfrac{1}{p^2-k^2}$ being a uniformly bounded operator we also infer that the same holds true for \eqref{eq:convergence:3:b}.
 
    It remains to estimate the contribution coming from the first term
    on the right  side of \eqref{eq:BS-decomposition}. With
    $P_\ell$ defined in \eqref{eq:P}, the norm of
    the corresponding vector in \eqref{eq:convergence:3:a} is
    \begin{equation}
      \label{eq:convergence:3:estimate}
      \begin{split}
        &\frac{1}{|e_\ell|}\Bigl\| \Bigl(\tfrac{1}{p^2-k^2} |V_\ell|^{1/2} - |g_k\rangle
        \bigl< |V_\ell|^{1/2}\bigr|\Bigr)
        P_\ell|V_\ell^{1/2}\rangle \Bigr\|_{L^2}\\
        &\quad= \left| \tfrac{\langle |V_\ell|^{1/2} |\phi_\ell \rangle} {e_\ell
          \langle J_\ell \phi_\ell | \phi_\ell\rangle} \right|
        \left(\int_{\mathbb{R}^3} \left|
            \int_{\mathbb{R}^3}|V_\ell|^{1/2}(z) \phi_\ell(z)
            \big(g_k(z-x) - g_k(x)\big) \ddd{3}z \right|^2
          \ddd{3}x\right)^{1/2} \\
        &\quad = \left| \tfrac{\langle |V_\ell|^{1/2} |\phi_\ell \rangle}
        {4 \pi e_\ell \langle J_\ell \phi_\ell| \phi_\ell\rangle} \right| \left(
          \int_{\mathbb{R}^6} \Omega_k(z,w)
          |V_\ell|^{1/2}(z)\phi_\ell(z) |V_\ell|^{1/2}(w)
          \phi_\ell(w) \ddd{3}z \ddd{3}w \right)^{1/2},
      \end{split}
    \end{equation}
    where
    \begin{align*}
      \Omega_k(z,w) & = (4\pi)^2 \int_{\mathbb{R}^3} [ g_{\bar k}(x-z) - g_{\bar k}(x) ]
      [g_k(x-w) -g_k(x) ] \ddd{3}x \\ &= F_k(w-z) + F_k(0) - F_k(w) - F_k(z) \,.
    \end{align*}
    The bound~\eqref{eq:G-difference} implies that 
    \begin{equation*}
      |\Omega_k(z,w)| \leq \frac 1{2\pi} \left(1+ \frac{|\Re (k)|}{2|\Im(k)|}\right) (|z|+|w|).
    \end{equation*}
    Together with Consequence~\ref{corollary:V-phi}\ref{corollary:V-phi:2} and
    \ref{corollary:V-phi:3}  we are thus able to estimate \eqref{eq:convergence:3:estimate} by $O(|e_\ell|^{1/2})$.
   This implies \eqref{eq:convergence:3:a}.
In order to get \eqref{eq:convergence:3:b} we make use of \eqref{eq:convergence:GV} and Consequence~\ref{corollary:V-phi}\ref{corollary:V-phi:2} and evaluate 
    \begin{align*}
      \| \tfrac{1}{p^2-k^2} |V_\ell|^{1/2} |\phi_\ell\rangle\|_{L^2}
      &\leq \| |g_k\rangle \langle |V_\ell|^{1/2}
      |\phi_\ell\rangle\|_{L^2} + \bigl\| \Bigl(\tfrac{1}{p^2-k^2}
      |V_\ell|^{1/2} -|g_k\rangle
      \langle |V_\ell|^{1/2}|\Bigr) |\phi_\ell\rangle \bigr\|_{L^2}\\
      &= |\langle |V_\ell| \big| \phi_\ell \rangle| + O(|e_\ell|^{1/2})
      = O(|e_\ell|^{1/2}).
    \end{align*}
    This completes the proof.
   \end{proof}

  Since the norm of
  \textcircled{\scriptsize 2} diverges in the limit of small $\ell$ we
  need to keep track of precise error bounds. In order to do that, we start by  rewriting the term  \textcircled{\scriptsize 2}  in a particularly useful way, which is presented in the following lemma. 

 \begin{lemma}
    \label{lemma:convergence:1}
    \begin{equation}
      \label{eq:convergence:1}
      \frac{1}{1+V_\ell^{1/2} \frac{1}{p^2-k^2}|V_\ell|^{1/2}} 
      = \Bigl(1-\tfrac{1}{\tfrac{4\pi}{ik}+4\pi
        a(V_\ell)} \frac{1}{1+ B_\ell} |V_\ell^{1/2}\rangle\langle |V_\ell|^{1/2}|\Bigr)
      \frac{1}{1+Q_\ell}
      \frac{1}{1+B_\ell},
    \end{equation}
    where
    \begin{align}\label{def:Q}
      Q_\ell = \frac{1}{1+B_\ell} R_\ell
      \left(1-\frac{1}{\tfrac{4\pi}{ik}+4\pi a(V_\ell)}
      \frac{1}{1+B_\ell} |V_\ell^{1/2}\rangle\langle
      |V_\ell|^{1/2}| \right)
    \end{align}
    and where $R_\ell$ is the operator with integral kernel
    \begin{equation}\label{def:R}
      R_\ell(x,y) = - \frac{ik}{4\pi}V_\ell^{1/2}(x)r(ik|x-y|)|V_\ell|^{1/2}(y)
    \end{equation}
    with $r(z) = \frac{\ee^z - 1 - z}{z}$.
  
    Moreover, $Q_\ell$ satisfies
    \begin{subequations}
      \begin{align}
        \label{eq:convergence:Q:1}
        \|Q_\ell\| &= O(|e_\ell|^{1/2}) \\
        \label{eq:convergence:Q:2}
        \left\|Q_\ell \frac{1}{1+B_\ell} V_\ell^{1/2}
        \frac{1}{p^2-k^2}\right\| &= O(|e_\ell|^{1/2}).
      \end{align}
    \end{subequations}
  \end{lemma}

  \begin{proof}
    The integral kernel of the operator $(p^2-k^2)^{-1}$ is given by $g_k(x-y)$. We expand this function as 
      \begin{equation*}
      g_k(x) = \frac{1}{4\pi} \frac{\ee^{ik|x|}}{|x|}
      =
      \frac{1}{4\pi}\frac{1}{|x|}
      +\frac{ik}{4\pi} - \frac{ik}{4\pi}r(ik|x|)
    \end{equation*}
    where $r(z) = \frac{\ee^z - 1 - z}{z}$. We insert this expansion in the expression for
    \textcircled{\scriptsize 2} and thus obtain the identity
  $$
      1+V_\ell^{1/2} \frac{1}{p^2-k^2}|V_\ell|^{1/2} =
     1+B_\ell + \frac{ik}{4\pi}
        |V_\ell^{1/2}\rangle\langle |V_\ell|^{1/2}| + R_\ell
        $$
         with  $R_\ell$ defined in (\ref{def:R}). We further rewrite this expression as 
        \begin{align*}
      &(1+B_\ell) \left(1+\frac{1}{1+B_\ell} \left(\frac{ik}{4\pi}
        |V_\ell^{1/2}\rangle\langle |V_\ell|^{1/2}| + R_\ell\right) 
        \right)
     \\      &= (1+B_\ell) 
           \left(
        1+\frac{1}{1+B_\ell} R_\ell
        \tfrac{1}{1+\frac{ik}{4\pi} \frac{1}{1+B_\ell}
          |V_\ell^{1/2}\rangle\langle |V_\ell|^{1/2}|}\right)
          \left(1+\frac{ik}{4\pi} \frac{1}{1+B_\ell}
        |V_\ell^{1/2}\rangle\langle |V_\ell|^{1/2}|\right)
\,.    \end{align*}
    The inverse of the operator in the last parenthesis can be  calculated explicitly, and is given by
    \begin{equation*}
    \left(1+\frac{ik}{4\pi} \frac{1}{1+B_\ell} |V_\ell^{1/2}\rangle\langle |V_\ell|^{1/2}|\right)^{-1}
      =
      1-\tfrac{1}{\tfrac{4\pi}{ik}+4\pi a(V_\ell)} \frac{1}{1+B_\ell}
      |V_\ell^{1/2}\rangle\langle |V_\ell|^{1/2}|,
    \end{equation*}
    at least whenever $a(V_\ell) \neq i/k$, which we can assume for small enough $\ell$. 
    Hence (\ref{eq:convergence:1}) holds, with $Q_\ell$ defined in (\ref{def:Q}). 
    
    We will now prove \eqref{eq:convergence:Q:1} and
    \eqref{eq:convergence:Q:2}.
    We have
    \begin{align*}
      \|Q_\ell \| &\leq \left\|\tfrac{1}{1+B_\ell} R_\ell\right\|
      +
      \tfrac{1}{\left|\tfrac{4\pi}{ik}+4\pi a(V_\ell)\right|}
      \left\|\tfrac{1}{1+B_\ell} R_\ell \tfrac{1}{1+B_\ell}
      |V_\ell^{1/2}\rangle\right\|_{L^2} \|V_\ell\|_{L^1}^{1/2} \\
      \left\|Q_\ell \tfrac{1}{1+B_\ell}
        V_\ell^{1/2}\tfrac{1}{p^2-k^2}
      \right\|
      &\leq \left\|\tfrac{1}{1+B_\ell} R_\ell \tfrac{1}{1+B_\ell}
      V_\ell^{1/2} \tfrac{1}{p^2-k^2}\right\|\\
      &\quad + \tfrac{1}{\left|\tfrac{4\pi}{ik}+4\pi a(V_\ell)\right|}
      \left\|
        \tfrac{1}{1+B_\ell} R_\ell \tfrac{1}{1+B_\ell}
        |V_\ell^{1/2}\rangle \right\|_{L^2}
      \left\| 
        \tfrac{1}{p^2-\bar{k}^2} |V_\ell|^{1/2} \tfrac{1}{1+B_\ell}
        |V_\ell^{1/2} \rangle \right\|_{L^2}.
    \end{align*}
  The norm $\|V_\ell\|_{L^1}$ is uniformly bounded by assumption, and it follows from Lemma~\ref{newlem} that also $\left\| 
      \tfrac{1}{p^2-\bar{k}^2} |V_\ell|^{1/2} \tfrac{1}{1+B_\ell}
      |V_\ell^{1/2} \rangle \right\|_{L^2}$ is uniformly bounded. Hence it
    suffices to bound the following expressions: 
    \begin{subequations}
      \begin{align}
        &\|\tfrac{1}{1+B_\ell} R_\ell\| \label{eq:Q:a} \\
        &\left\|\tfrac{1}{1+B_\ell} R_\ell \tfrac{1}{1+B_\ell}
          |V_\ell^{1/2}\rangle \right\|_{L^2} \label{eq:Q:b}\\
        &\left\| \tfrac{1}{1+B_\ell} R_\ell \tfrac{1}{1+B_\ell} V_\ell^{1/2} \tfrac{1}{p^2-k^2}
        \right\|. \label{eq:Q:c}
      \end{align}
    \end{subequations}
    To this aim we again use the decomposition \eqref{eq:BS-decomposition}.    For \eqref{eq:Q:a} note that
    \begin{equation}
      \label{eq:RB-bound:1}
      \begin{split}
        \| R_\ell \|_2^2 &= \frac{|k|^2}{(4\pi)^2}\int_{\mathbb{R}^6}
        |V_\ell(x)| |r(ik|x-y|)|^2 |V_\ell(y)|
        \ddd{3}x \ddd{3}y\\
        &\leq \frac {c^2 |k|^2}{(4\pi)^2} \int_{\mathbb{R}^6} |V_\ell(x)| |x-y|^2 |V_\ell(y)|
        \ddd{3}x \ddd{3}y\\
        &\leq  \frac{4 c^2 |k|^2}{(4\pi)^2} \|V_\ell\|_{L^1} \int_{\mathbb{R}^3}
        |V_\ell(x)|\, |x|^2 \ddd{3}x = O(e_\ell^2),
      \end{split}
    \end{equation}
       using  $|r(z)| \leq c |z|$ for $\Re z<0$ and some $c>0$, as well as Assumptions~\ref{ax:V} and  \ref{ax:V-x}. Since the last term in (\ref{eq:BS-decomposition}) is uniformly bounded by Consequence \ref{corollary:projection-bound},
    $\|\tfrac{1}{1+B_\ell}(1-P_\ell)R_\ell\|_2 =
    O(e_\ell)$ and $\|R_\ell \tfrac{1}{1+B_\ell}(1-P_\ell)\|_2 =
    O(e_\ell)$.  On the other hand
    \begin{equation*}
     \| P_\ell R_\ell\|
      = \frac{\|R^*_\ell J_\ell
        \phi_\ell\|_{L^2}}{|\langle J_\ell \phi_\ell| \phi_\ell\rangle|},
    \end{equation*}
    where
    \begin{equation*}
      (R^*_\ell J_\ell \phi_\ell)(x)
      =
      \frac{i \bar k}{4\pi}\int_{\mathbb{R}^3} |V_\ell(x)|^{1/2}r(-i \bar k|x-y|)
      |V_\ell(y)|^{1/2} \phi_\ell(y) \ddd{3}y.
    \end{equation*}
    Using again the above pointwise bound on $r$, it follows from Consequence~\ref{corollary:V-phi} that
    \begin{equation}
      \label{eq:RB-bound:2}
    \| P_\ell R_\ell\|
      = O(|e_\ell|^{3/2})\,.
    \end{equation}
    Thus $\eqref{eq:Q:a} = O(|e_\ell|^{1/2})$.
  
    Finally, we estimate \eqref{eq:Q:b} and \eqref{eq:Q:c}.
 Proceeding as above one also shows that $\|R_\ell \tfrac{1}{1+B_\ell}\| = O(|e_\ell|^{1/2})$.  Hence, in the decomposition \eqref{eq:BS-decomposition}
    for the vector
    \begin{equation*}
      \frac{1}{1+B_\ell} R_\ell \frac{1}{1+B_\ell}
      |V_\ell^{1/2}\rangle
    \end{equation*}
    and the operator
    \begin{equation*}
      \frac{1}{1+B_\ell} R_\ell \frac{1}{1+B_\ell}
      V_\ell^{1/2} \frac{1}{p^2-k^2},
    \end{equation*}
    we see that the only parts left to estimate are
    \begin{subequations}
      \begin{align}
        \label{eq:remainder:a}
        \frac{1}{e_\ell^2} P_\ell R_\ell P_\ell |V_\ell^{1/2}\rangle &=
        \frac{1}{e_\ell^2} \langle J_\ell\phi_\ell | R_\ell \phi_\ell
        \rangle \frac{\bigl< |V_\ell|^{1/2} \big| \phi_\ell\bigr>}
        {|\langle J_\ell \phi_\ell | \phi_\ell\rangle|^2} |\phi_\ell\rangle,\\
        \label{eq:remainder:b}
        \frac{1}{e_\ell^2} P_\ell R_\ell P_\ell V_\ell^{1/2}
        \frac{1}{p^2-k^2} &= \frac{1}{e_\ell^2} \langle
        J_\ell\phi_\ell | R_\ell \phi_\ell \rangle
        \frac{|\phi_\ell\rangle\langle \phi_\ell|
          |V_\ell|^{1/2}\frac{1}{p^2-k^2}} {|\langle J_\ell \phi_\ell
          | \phi_\ell\rangle|^2}.
      \end{align}
    \end{subequations}
    Using again Consequence~\ref{corollary:V-phi} and the pointwise bound on $r$ 
   we obtain $ |\langle J_\ell\phi_\ell | R_\ell \phi_\ell \rangle |= O(e_\ell^2)$.
   In particular, 
    we conclude that  the $L^2$-norm of the vector in \eqref{eq:remainder:a} is of order
    $O(|e_\ell|^{1/2})$.  The same argument  applies to the operator in
    \eqref{eq:remainder:b} after using \eqref{eq:convergence:GV}. 
    This shows that \eqref{eq:Q:b} and \eqref{eq:Q:c} are of order
    $O(|e_\ell|^{1/2})$, and completes the proof.
    \end{proof}

  The estimates \eqref{eq:convergence:Q:1} and
  \eqref{eq:convergence:Q:2} suggest that for small $\ell$ we may drop $Q_\ell$ in \textcircled{\scriptsize 2} in 
  \eqref{eq:resolvent}.  With the help of the identity \eqref{eq:convergence:1} and the expansion
  $\tfrac{1}{1+Q_\ell} = 1 - \frac{1}{1+Q_\ell} Q_\ell$ the second
  summand on the right  side of \eqref{eq:resolvent}  decomposes
  into two parts, namely
  \begin{align*}
  - \frac{1}{p^2-k^2}|V_\ell|^{1/2}
    \frac{1}{1+V_\ell^{1/2}\frac{1}{p^2-k^2}|V_\ell|^{1/2}}
    V_\ell^{1/2} \frac{1}{p^2-k^2} = {\rm I}_\ell + {\rm I\!I}_\ell,
  \end{align*}
  with
  \begin{align*}
    {\rm I}_\ell &= -
    \frac{1}{p^2-k^2}|V_\ell|^{1/2}\Bigl(1-\tfrac{1}{\tfrac{4\pi}{ik}+4\pi
      a(V_\ell)}\frac{1}{1+ B_\ell} |V_\ell^{1/2}\rangle\langle
    |V_\ell|^{1/2}|\Bigr)
    \frac{1}{1+B_\ell}V_\ell^{1/2} \frac{1}{p^2-k^2},\\
    {\rm I\!I}_\ell &=
    \frac{1}{p^2-k^2}|V_\ell|^{1/2}\Bigl(1-\tfrac{1}{\tfrac{4\pi}{ik}+4\pi
      a(V_\ell)}\frac{1}{1+ B_\ell} |V_\ell^{1/2}\rangle\langle
    |V_\ell|^{1/2}|\Bigr) \frac{1}{1+Q_\ell} Q_\ell \frac{1}{1+B_\ell}V_\ell^{1/2} \frac{1}{p^2-k^2}.
  \end{align*}
  The term $ {\rm I}_\ell$ contains the main part, whereas ${\rm I\!I}_\ell$ vanishes in operator norm for small $\ell$.
  This is the content of the
   following lemma, which  immediately implies the statement of Theorem~\ref{thm:convergence}. Its proof relies heavily on Lemmas~\ref{newlem} and~\ref{lemma:convergence:1}. 

  \begin{lemma}
    \label{lemma:convergence:2}
    \begin{subequations}
      \begin{align}
        \bigl\|{\rm I}_\ell + \tfrac{4\pi}{a(V_\ell)^{-1}+ ik }
        |g_k\rangle \langle g_k| \bigr\| &= O(|e_\ell|^{1/2})
        \label{eq:convergence:2:a}\\
        \| {\rm I\!I}_\ell \| &=
        O(|e_\ell|^{1/2}).\label{eq:convergence:2:b}
      \end{align}
    \end{subequations}
  \begin{proof}
    We first show \eqref{eq:convergence:2:a}. We can write
    \begin{align*}
      {\rm I}_\ell &= -
      \frac{1}{p^2-k^2}|V_\ell|^{1/2}\frac{1}{1+B_\ell}V_\ell^{1/2} \frac{1}{p^2-k^2} \\
      &\quad + \tfrac{1}{\tfrac{4\pi}{ik}+4\pi a(V_\ell)}
      \frac{1}{p^2-k^2}|V_\ell|^{1/2}\frac{1}{1+ B_\ell}
      |V_\ell^{1/2}\rangle\langle |V_\ell|^{1/2}| \frac{1}{1+B_\ell}V_\ell^{1/2} \frac{1}{p^2-k^2}.
    \end{align*}
    It turns out that both summands converge in operator norm to the 
    projector $|g_k\rangle \langle g_k|$,  multiplied by numbers which add
    up to $-\tfrac{-4\pi}{a^{-1}+ ik }$. 
    More precisely, we are
    going derive to the following asymptotic behavior  
    \begin{subequations}
      \begin{equation}
        \label{eq:convergence:4:a}
        \begin{split}
          &\Bigl\|\tfrac{1}{p^2-k^2} |V_\ell|^{1/2} \tfrac{1}{1+B_\ell} V_\ell^{1/2} \tfrac{1}{p^2-k^2} - 4\pi a(V_\ell)
          |g_k\rangle \langle g_k| \Bigr\| \\
          &\quad= \Bigl\|\tfrac{1}{p^2-k^2} |V_\ell|^{1/2}
          \tfrac{1}{1+B_\ell} V_\ell^{1/2} \tfrac{1}{p^2-k^2} -
          |g_k\rangle \bigr< |V_\ell|^{1/2} \big| \tfrac{1}{1+B_\ell}
          \big|V_\ell^{1/2}\bigr>\langle g_k| \Bigr\|\\
          &\quad\leq \Bigl\|\Bigl(\tfrac{1}{p^2-k^2} |V_\ell|^{1/2} -
          |g_k\rangle \bigr< |V_\ell|^{1/2} \big|\Bigr) \tfrac{1}{1+B_\ell} V_\ell^{1/2} \tfrac{1}{p^2-k^2}\Bigr\|\\
          &\qquad+ \Bigl\| |g_k\rangle \bigr< |V_\ell|^{1/2} \big|
          \tfrac{1}{1+B_\ell} \Bigl( V_\ell^{1/2}
          \tfrac{1}{p^2-k^2} - \big|V_\ell^{1/2}\bigr>\langle
          g_k|\Bigr) \Bigr\|\\
          &\quad = O(|e_\ell|^{1/2})
        \end{split}
      \end{equation}
      and
      \begin{equation}
        \label{eq:convergence:4:b}
        \begin{split}
          &\Bigl\|\tfrac{1}{p^2-k^2}|V_\ell|^{1/2}\tfrac{1}{1+ B_\ell} |V_\ell^{1/2}\rangle\langle |V_\ell|^{1/2}|
          \tfrac{1}{1+B_\ell}V_\ell^{1/2} \tfrac{1}{p^2-k^2}
          -\big(4\pi a(V_\ell)\big)^2 |g_k\rangle \langle
          g_k|
          \Bigr\| \\
          &\quad\leq \Bigl\|\Bigl( \tfrac{1}{p^2-k^2}|V_\ell|^{1/2}
          -|g_k\rangle\bigr< |V_\ell|^{1/2} \big| \Bigr) \tfrac{1}{1+ B_\ell} |V_\ell^{1/2}\rangle\langle |V_\ell|^{1/2}|
          \tfrac{1}{1+B_\ell}V_\ell^{1/2} \tfrac{1}{p^2-k^2}
          \Bigr\|\\
          &\qquad+ 4\pi |a(V_\ell)| \Bigl\| |g_k\rangle\langle
          |V_\ell|^{1/2}| \tfrac{1}{1+B_\ell} \Bigl(
          V_\ell^{1/2} \tfrac{1}{p^2-k^2} - \big|V_\ell^{1/2}\bigr> \langle
          g_k|\Bigr)
          \Bigr\|\\
          &\quad = O(|e_\ell|^{1/2}),
        \end{split}
      \end{equation}
\end{subequations}
      where we made use of the expression $a(V_\ell) = \frac 1{4\pi} \bigr< |V_\ell|^{1/2} \big| \tfrac{1}{1+B_\ell}
          \big|V_\ell^{1/2}\bigr>$ for the scattering length.
 The bounds \eqref{eq:convergence:4:a} and \eqref{eq:convergence:4:b} are in fact simple consequences of Lemma~\ref{newlem}. Eq.~\eqref{eq:convergence:4:a} follows immediately from \eqref{eq:convergence:3:a} and \eqref{eq:convergence:3:b}. To see \eqref{eq:convergence:4:b} we apply \eqref{eq:convergence:3:a} twice, once to the first vector in the first term on the right side, and once to the second.
 
     In order to show
    \eqref{eq:convergence:2:b} we simply  bound ${\rm I\!I}_\ell$ by
    \begin{align*}
      \|{\rm I\!I}_\ell \| &\leq \Bigl\|
      \tfrac{1}{p^2-k^2}|V_\ell|^{1/2}\Bigl(1-\tfrac{1}{\tfrac{4\pi}{ik}+4\pi
        a(V_\ell)}\tfrac{1}{1+ B_\ell}
      |V_\ell^{1/2}\rangle\langle |V_\ell|^{1/2}|\Bigr)
      \frac{1}{1+Q_\ell} \Bigr\| \Bigl\| Q_\ell
      \tfrac{1}{1+B_\ell}V_\ell^{1/2}
      \tfrac{1}{p^2-k^2} \Bigr\|\\
      &\leq O(|e_\ell|^{1/2})\,,
    \end{align*}
    where we used that the first term is uniformly bounded because of \eqref{eq:convergence:3:a}, whereas the second term vanishes like $O(|e_\ell|^{1/2})$
     thanks to \eqref{eq:convergence:Q:2}.
  \end{proof}
   \end{lemma}

\section{Proof of Consequences~\ref{corollary:projection-bound} and \ref{corollary:V-phi}}
\label{sec:proof-cor}

\begin{proof}[Proof of Consequence~\ref{corollary:projection-bound}]
 We pick some $\psi\in L^2(\mathbb{R}^3)$ and set
  \begin{equation}
    \label{eq:varphi}
    \varphi = \frac{1}{1+J_\ell
      X_\ell}(1-P_\ell)\psi = \frac{1}{J_\ell+
      X_\ell}J_\ell(1-P_\ell)\psi\,.
  \end{equation}
  Below, we are going to show that there exists a constant $c>0$ such that for small
  enough $\ell$
  \begin{equation}
    \label{eq:aah}
    \langle \varphi | (1-X_\ell^-) \varphi \rangle 
    \geq c \|\varphi\|^2_{L^2}.
  \end{equation}
In combination with Lemma~\ref{lemma:BS-estimate} this inequality  yields 
  \begin{align*}
    \sqrt{2} \|\varphi\|\|(J_\ell+X_\ell)\varphi\| &\geq \langle
    \varphi| (1-X_\ell^-) \varphi \rangle \geq c \|\varphi\|^2,
  \end{align*}
  which further implies that
  \begin{equation*}
    \|\psi\| \geq \|J_\ell(1-P_\ell)\psi\|
    = \|(J_\ell+X_\ell)\varphi\|
    \geq
    \frac c {\sqrt{2}} \|\varphi\| = \frac c {\sqrt{2}} \|(1+B_\ell)^{-1}(1-P_\ell)\psi\|\,,
  \end{equation*}
  proving the statement.
  
  It remains to show the inequality \eqref{eq:aah}. 
  To this aim we denote by $\phi_\ell^-$ the eigenvector corresponding to the smallest eigenvalue  $e_\ell^-$ of $1-X_\ell^-$ and by $P_{\phi_\ell^-}$ the orthogonal
  projection onto $\phi_\ell^-$.  By assumption, the Birman-Schwinger operator
  $X_\ell^-$ corresponding to the potential $V_\ell^-$ has only one
  eigenvalue close to $1$. All other
  eigenvalues are separated from $1$ by a gap of order one. Hence
  there exists $c_1 > 0$ such that
  \begin{equation*}
    (1-X_\ell^-)(1-P_{\phi_\ell^-}) \geq c_1
  \end{equation*}
  and, therefore,
  \begin{align*}
    \langle \varphi | (1-X_\ell^-) \varphi \rangle
    &\geq c_1 \langle \varphi | (1-P_{\phi^-_\ell}) \varphi \rangle 
    + e_\ell^- \langle \varphi | P_{\phi^-_\ell} \varphi \rangle\\
    &= c_1 \| \varphi \|_{L^2}^2 + (e_\ell^- - c_1)
    \langle \varphi | P_{\phi^-_\ell} \varphi \rangle.
  \end{align*}
 With $P_{J_\ell \phi_\ell} = | J_\ell \phi_\ell\rangle \langle
  J_\ell \phi_\ell|$ being the orthogonal projection onto
  $J_\ell\phi_\ell$ we can write 
  \begin{equation*}
    \varphi = (1-P_{J_\ell\phi_\ell})\varphi,
  \end{equation*}
 simply for the reason that, because of \eqref{eq:varphi} and the fact that $P_\ell$ commutes with $B_\ell$,
  \begin{align*}
P_{J_\ell \phi_\ell} \varphi =     P_{J_\ell \phi_\ell} (1+B_\ell)^{-1}(1-P_\ell) \psi= P_{J_\ell
      \phi_\ell} (1-P_\ell) (1+B_\ell)^{-1} \psi= 0\,.
  \end{align*}
  Consequently,
  \begin{align*}
    |\langle \varphi | P_{\phi^-_\ell} \varphi \rangle|
    &=
    |\langle \varphi | (1-P_{J_\ell\phi_\ell}) P_{\phi^-_\ell} \varphi \rangle|
    \leq
    \|\varphi\|^2_{L^2} \|(1-P_{J_\ell\phi_\ell})P_{\phi^-_\ell} \|\\
    &=
    \|\varphi\|^2_{L^2}\|(1-P_{J_\ell\phi_\ell})\phi_\ell^-\|^2
    = 
    \|\varphi\|^2_{L^2} \|(1-P_{\phi_\ell^-})J_\ell\phi_\ell\|^2\,.
  \end{align*}
  To estimate $\|(1-P_{\phi_\ell^-})J_\ell\phi_\ell\|$, we apply
  Lemma~\ref{lemma:BS-estimate} to $\phi_\ell$ and obtain
  \begin{align*}
    \sqrt{2}\, e_\ell = \sqrt{2}\, \|(J_\ell + X_\ell)\phi_\ell\|
    &\geq \langle \phi_\ell|(1-X_\ell^-)\phi_\ell\rangle
    = \langle J_\ell \phi_\ell|(1-X_\ell^-)J_\ell \phi_\ell\rangle \\
    &= e^-_\ell|\langle J_\ell\phi_\ell|\phi_\ell^-\rangle|^2 +
    \langle (1-P_{\phi_\ell^-})J_\ell
    \phi_\ell|(1-X_\ell^-)(1-P_{\phi_\ell^-})J_\ell \phi_\ell\rangle\\
    &\geq
    e^-_\ell|\langle J_\ell\phi_\ell|\phi_\ell^-\rangle|^2
    +c_1\|(1-P_{\phi_\ell^-})J_\ell \phi_\ell\|^2\,.
  \end{align*}
  This shows that 
  $\|(1-P_{J_\ell\phi_\ell})P_{\phi^-_\ell} \| = O(|e_\ell|^{1/2})$ and consequently \eqref{eq:aah} holds for small enough $\ell$.
\end{proof}

\begin{proof}[Proof of Consequence~\ref{corollary:V-phi}]
(i) Simply bound $|x|\leq \frac12 (|e_\ell| + |x|^2/|e_\ell|)$ and use 
Assumptions \ref{ax:V} and  \ref{ax:V-x}.

(ii)  Lemma~\ref{lemma:BS-estimate} applied to $\phi_\ell$ implies that 
  \begin{equation}
    \label{eq:lim:X}
    \bigl< \phi_\ell\big| (V_\ell^+)^{1/2}
    \tfrac{1}{p^2}(V_\ell^+)^{1/2} \phi_\ell\bigr> \leq \sqrt{2}
    |e_\ell| + |e_\ell^-|, \qquad
    \bigl< \phi_\ell\big| \big(1- (V_\ell^-)^{1/2}
    \tfrac{1}{p^2}(V_\ell^-)^{1/2}\big) \phi_\ell\bigr> \leq \sqrt{2} |e_\ell|\,,
  \end{equation}
  where we used $1- (V_\ell^-)^{1/2}\frac{1}{p^2}(V_\ell^-)^{1/2} \geq
  e_\ell^-$. Note, that by \ref{ax:e} $|e_\ell^-| = O(e_\ell)$.  Now
  \ref{corollary:V-phi:1} follows from the following argument.  Because of  \eqref{eq:BS-estimate:2} we know that
  \begin{align*}
    \sqrt{2}|e_\ell | &= \sqrt{2}\|(J_\ell+X_\ell)\phi_\ell \|_{L^2} \geq
    \langle \phi_\ell^+| (1+X_\ell^+) \phi_\ell^+\rangle
    + \langle \phi_\ell^-| (1-X_\ell^-) \phi_\ell^-\rangle\\
    &= \| \phi_\ell^+ \|_{L^2}^2 + \langle \phi_\ell^+| X_\ell^+
    \phi_\ell^+\rangle + \langle \phi_\ell^-| (1-X_\ell^-) \phi_\ell^-
    \rangle,
  \end{align*}
  where $\phi_\ell = \phi_\ell^+ + \phi_\ell^-$ with
  $\supp(\phi_\ell^-) \subseteq \supp(V_\ell^-)$ and
  $\supp(\phi_\ell^+) \cap \supp(V_\ell^-) = \emptyset$.  Using that
  $1-X_\ell^-$ has $e_\ell^- = O(e_\ell)$ as lowest eigenvalue, we
  conclude that
  \begin{equation}
    \label{eq:J:expectation}
    \bigl< \tfrac{1+J_\ell}{2} \phi_\ell\big|\phi_\ell\bigr> = \| \phi_\ell^+ \|_{L^2}^2 = O(e_\ell).
  \end{equation}

(iii),(iv)  For $q=0,1$ we
  evaluate  
  \begin{align*}
    \int_{\mathbb{R}^3} |x|^q\,|V_\ell|^{1/2}\,|\phi_\ell|\ddd{3}x &=
    \bigl< |V^+_\ell|^{1/2} |\cdot|^q\big| |\phi_\ell|\bigr> + \bigl<
    (V^-_\ell)^{1/2} |\cdot|^q \big| |\phi_\ell|
    \bigr>\\
    &= \bigl< |V^+_\ell|^{1/2}|\cdot|^q \big|
    \frac{1}{2}(1+J_\ell)|\phi_\ell|\bigr> +
    \bigl<  (V^-_\ell)^{1/2}|\cdot|^q \big| |\phi_\ell| \bigr>\\
    &\leq \|V^+_\ell |\cdot|^{2q} \|_{L^1}^{1/2} \langle \phi_\ell |
    \tfrac{1}{2}(1+J_\ell)\phi_\ell\rangle^{1/2} + \||\cdot|^{2q}
    V_\ell^-\|_{L^1}^{1/2},
  \end{align*}
  which is $O(|e_\ell|^{1/2})$ for $q=0$ and $O(|e_\ell|^{3/2})$ for $q=1$
  by assumption \ref{ax:V-x} and \eqref{eq:J:expectation}.
\end{proof}

\end{document}